\documentclass[11pt]{article}
\usepackage[english]{babel}
\usepackage[utf8]{inputenc}
\usepackage{amsmath,amsfonts,amsthm,units,mathtools} % Math packages
\usepackage[dvips]{graphics}
\usepackage{xcolor}
\usepackage{epsfig}
\usepackage{url,lmodern}
\usepackage{amsmath,amssymb,indentfirst,epsfig,bm,bbm}

\newcommand{\vn}{{\bf v}}

\newcommand{\lamn}{{\mbox{\boldmath $\lambda$}}}

\newcommand{\etn}{{\mbox{\boldmath $\eta$}}}

\title{\huge {An extended Rayleigh model: Properties,
regression and COVID-19 application}}

\author{{ Gauss M. Cordeiro} \\
Universidade Federal de Pernambuco\\%, Brazil.E-mail:
\url{gausscordeiro@gmail.com} \\[0.2cm]
{ Gabriela M. Rodrigues} \\
Universidade de S\~ao Paulo \\ %, Brazil. E-mail:
\url{gabrielar@usp.br} \\[0.2cm]
{ Edwin M. M. Ortega} \\
 Universidade de S\~ao Paulo \\%, Brazil. E-mail:
\url{edwin@usp.br} \\[0.2cm]
{ Lu\'is H. de Santana} \\
Universidade Federal de Pernambuco \\ %, Brazil. E-mail:
\url{desantanalh@gmail.com}\\[0.2cm]
{ Roberto Vila} \\
Universidade de Bras\'ilia \\ %, Brazil. E-mail:
\url{rovig161@gmail.com}}
\date{}

\theoremstyle{thmstyleone}%
\newtheorem{theorem}{Theorem}%  meant for continuous numbers
\newtheorem{proposition}[theorem]{Proposition}%

\theoremstyle{thmstyletwo}%

\theoremstyle{thmstylethree}%
\begin{document}
\maketitle
\begin{abstract}
We define a four-parameter extended Rayleigh distribution,
and obtain several mathematical properties including a stochastic representation. We construct a regression from the new distribution. The estimation is done by maximum likelihood. The utility of the new models is proved in two real applications. \\
\\
\vspace{1mm} \noindent {\bf Keywords}: Censored data,
generalized Rayleigh; maximum likelihood estimation; regression model, stochastic representation.
\end{abstract}

\section{Introduction}

The Rayleigh distribution has been
employed in many areas (Johnson {\it et al.}, 1994).
Several of its extensions have been published so far.

The cumulative distribution function (cdf) of the
generalized Rayleigh (GR) distribution (Vod\u{a}, 1976) is
\begin{eqnarray}\label{cdfGR}
G_{\text{GR}}(x;\delta,\theta)=\gamma_1(\delta+1,\theta x^2),\quad x>0,
\end{eqnarray}
where $\delta>-1$ and $\theta>0$,
$\Gamma(p)=\int_{0}^{\infty}w^{p-1}\,{\rm{e}}^{-w}dw$
and
$\gamma_1(p,x)=\Gamma(p)^{-1}\int_{0}^{x}w^{p-1}\,
{\rm{e}}^{-w}d w$ are the gamma function and lower
incomplete gamma function ratio, respectively.

The GR distribution includes well-known sub-models.
The classical Rayleigh
distribution follows when $\delta=0$ and $\theta=\lambda^{-2}$.
If $\delta=2^{-1}$ and $\theta=(2\lambda^2)^{-1}$, it gives
the Maxwell distribution. The chi-square refers to $\theta=(2\tau^{2})^{-1}$, $\tau>0$, and
$\delta=\nicefrac{n}{2}-1$, $n\in \mathbb{N}$,
and the half-normal to $\delta=-2^{-1}$ and $\theta={2\sigma}^{-2}$.

The probability density function (pdf)
corresponding to (\ref{cdfGR}) is
\begin{eqnarray}\label{pdfGR}
g_{\text{GR}}(x;\delta,\theta)
=\frac{2 \theta^{\delta+1}}{\Gamma(\delta+1)}\,x^{2 \delta+1}\,{\rm{e}}^{- \theta x^2}.
\end{eqnarray}
Let $Z\sim \text{GR}(\delta,\theta)$ be a random variable with density (\ref{pdfGR}). The moments of $Z$ are easily obtained
from the integral given in Section 3.478 of Gradshteyn and Ryzhik (2000) $\int_{0}^\infty x^{\nu-1} {\rm{e}}^{-\mu x^p} dx =
\nicefrac{\Gamma\left(\nicefrac{\nu}{p}\right)}{p\,\mu^{\nicefrac{\nu}{p}}}$, where $p,\nu,\mu>0$.
Indeed, the $s$th ordinary moment of $Z$ (for a positive real number $s$) is
\begin{equation}\label{momentGR}
\mu_{s}^{\prime}(\delta, \theta)
=\frac{\Gamma(\nicefrac{s}{2}+\delta+1)}{\theta^{\nicefrac{s}{2}}\,\Gamma(\delta+1)}.
\end{equation}
%{\color{red} What other results should I insert?!}

The cdf of the {\it generalized odd log-logistic-G}
(GOLL-G) family follows from Cordeiro {\it et al.} (2017)
\begin{eqnarray}\label{cdfF}
F(x;\alpha,\beta,\boldsymbol{\xi})=\frac{G(x;\boldsymbol{\xi})^{\alpha\beta}}
{G(x;\boldsymbol{\xi})^{\alpha\beta}+\left[1-G(x;\boldsymbol{\xi})^{\beta}\right]^{\alpha}},
\end{eqnarray}
where $\alpha>0$ and $\beta>0$ are two extra parameters.

The generalized log-logistic (Gleaton and Lynch, 2006)
and proportional reversed hazard rate (Gupta and Gupta, 2007)
are special models of Equation~\eqref{cdfF}. Further, the
parent comes when $\alpha=\beta=1$. Prataviera {\it et al.} (2018) discussed the generalized odd log-logistic flexible Weibull distribution.

If $g(x;\boldsymbol{\xi})$ is the parent density,
the pdf corresponding to~\eqref{cdfF} can be written as
\begin{eqnarray}\label{pdfF}
f(x;\alpha,\beta,\boldsymbol{\xi})
= \frac{\alpha\beta g(x;\boldsymbol{\xi}) \,  G(x;\boldsymbol{\xi})^{\alpha\beta-1}[1-G(x;\boldsymbol{\xi})^{\beta}]^{\alpha-1}}{\big\{G(x;\boldsymbol{\xi})^{\alpha\beta}
+\left[1-G(x;\boldsymbol{\xi})^{\beta}\right] ^{\alpha}\big\}^{2}}\, .
\end{eqnarray}

The paper is structured as follows. Section \ref{sec:gollgr} defines the {\it generalized odd log-logistic generalized Rayleigh} (GOLLGR) distribution, addresses some asymptotes
and quantile function (qf),  gives a linear
representation for the family density, reports
moments and generating function, and addresses maximum likelihood estimation. A new regression is constructed in Section \ref{sec:regression}, and some simulation studies are
carried out as well as residual analysis.
Two lifetime data sets in Section \ref{sec:applicatons}
show the utility of the new models. Some conclusions are
reported in Section \ref{sec:conclusoes}.

\section{The GOLLGR model, properties and estimation}\label{sec:gollgr}

The GOLLGR density is defined (for $x >0$) by inserting (\ref{cdfGR}) and (\ref{pdfGR}) in Equation (\ref{pdfF})
\begin{eqnarray}\label{newpdf}
f(x)=
f(x;\alpha,\beta,\delta,\theta)
=
\frac{
2 \alpha\beta \theta^{\delta+1}\,x^{2 \delta+1}\,{\rm{e}}^{- \theta x^2}\,\gamma_1\left(\delta+1,\theta x^2\right)^{\alpha\beta-1}\big[1-\gamma_1\left(\delta+1,\theta x^2\right)^{\beta}\big]^{\alpha-1}
}
{
\Gamma(\delta+1)\big\{\gamma_1\left(\delta+1,\theta x^2\right)^{\alpha\beta}
+\big[1-\gamma_1\left(\delta+1,\theta x^2\right)^{\beta}\big] ^{\alpha}\big\}^{2}
},
\end{eqnarray}
and its hazard rate function (hrf) is
\begin{eqnarray}\label{newhrf}
h(x)
=
\frac{2 \alpha\beta \theta^{\delta+1}
\,x^{2 \delta+1}\,{\rm{e}}^{- \theta x^2}
\,\gamma_1\left(\delta+1,\theta x^2\right)^{\alpha\beta-1}
}{\Gamma(\delta+1)
\,\big[1-\gamma_1\left(\delta+1,
\theta x^2\right)^{\beta}\big] \big\{\gamma_1\left(\delta+1,
\theta x^2\right)^{\alpha\beta}
+\big[1-\gamma_1\left(\delta+1,\theta x^2\right)^{\beta}
\big]^{\alpha}\big\}}\, \cdot
\end{eqnarray}

We have $\lim_{x\to \infty}f(x)=0$, and
\begin{eqnarray}\label{limit-f}
\lim_{x\to 0^+} f(x)
=
\begin{cases}
{2\sqrt{\theta}\over\sqrt{\pi}}, & \alpha\beta-1= 0, \ 2 \delta+1= 0,
\\[0,3cm]
0, & \alpha\beta-1> 0, \ 2 \delta+1> 0,
\\[0,3cm]
\infty, & \alpha\beta-1> 0, \ 2 \delta+1< 0, \ 0<\delta+1<{2\delta+1\over 2(1-\alpha\beta)},
\\[0,3cm]
{2\alpha\beta\sqrt{\theta}\over\Gamma(\delta+1)}\,
\Bigl[{2(1-\alpha\beta) \over \Gamma(\delta+1)(2\delta+1) }\Bigr]^{\alpha\beta-1}, & \alpha\beta-1> 0, \ 2 \delta+1< 0, \ \delta+1={2\delta+1\over 2(1-\alpha\beta)},
\\[0,3cm]
0, & \alpha\beta-1> 0, \ 2 \delta+1< 0, \ \delta+1>{2\delta+1\over 2(1-\alpha\beta)},
\\[0,3cm]
%\infty, & \alpha\beta-1= 0, \ 2 \delta+1< 0,
%\\[0,3cm]
\infty, & \alpha\beta-1< 0, \ 2 \delta+1> 0, \ \delta+1>{2\delta+1\over 2(1-\alpha\beta)},
\\[0,3cm]
{2\alpha\beta\sqrt{\theta}\over\Gamma(\delta+1)}\,
\Bigl[{2(1-\alpha\beta) \over \Gamma(\delta+1)(2\delta+1) }\Bigr]^{\alpha\beta-1}, & \alpha\beta-1< 0, \ 2 \delta+1> 0, \ \delta+1={2\delta+1\over 2(1-\alpha\beta)},
\\[0,3cm]
0, & \alpha\beta-1< 0, \ 2 \delta+1> 0,
\ 0<\delta+1<{2\delta+1\over 2(1-\alpha\beta)},
\\[0,3cm]
\infty, & \alpha\beta-1\leqslant 0, \ 2 \delta+1\leqslant 0, \ \alpha\beta-1\neq 0  \ \text{or} \ 2 \delta+1\neq 0.
\end{cases}
\end{eqnarray}

Further, $\lim_{x\to \infty} h(x)=\infty$ and $\lim_{x\to 0^+} h(x)=\lim_{x\to 0^+} f(x)$.

Hereafter, we omit the arguments of the functions, and let
$X \sim \mbox{GOLLGR}(\alpha,\beta,\delta,\theta)$
be a random variable with pdf \eqref{newpdf},
The odd log-logistic GR follows when $\beta=1$, and the proportional reversed hazard rate GR when $\alpha=1$.
Plots of the pdf and hrf of $X$ are displayed in Figures\ref{gollgrpdf} and \ref{gollgrfrf}, respectively.
They reveal the bimodality and asymmetry of the pdf, and
different shapes of the hrf.

\begin{figure}[h]
\centering		
\includegraphics[width=17cm]{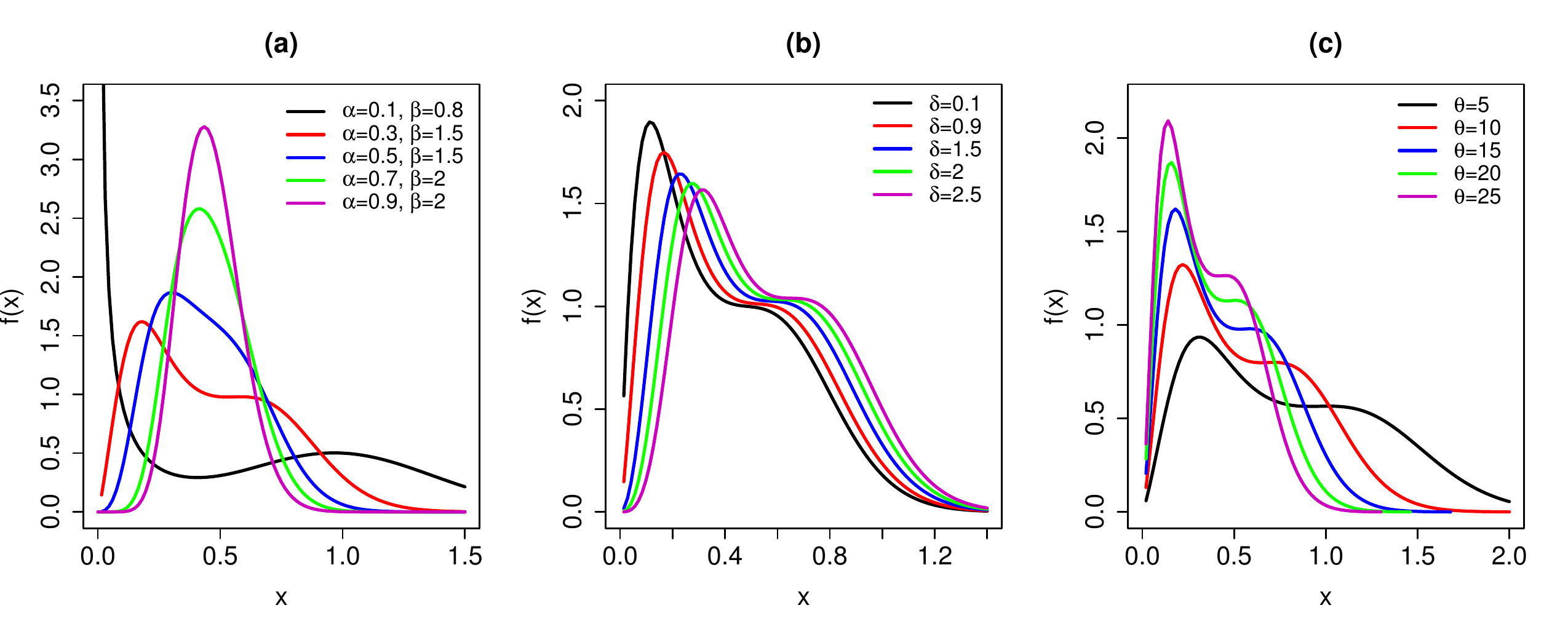}
	\caption{The GOLLGR pdf. (a) $\delta=1.5$ and $\theta=15$. (b) $\alpha=0.3$, $\beta=2$ and $\theta=15$. (c) $\alpha=0.3$, $\beta=1.5$ and $\delta=1.5$.}
	\label{gollgrpdf}

\end{figure}

	\begin{figure}[h]
\centering		\includegraphics[width=17cm]
		{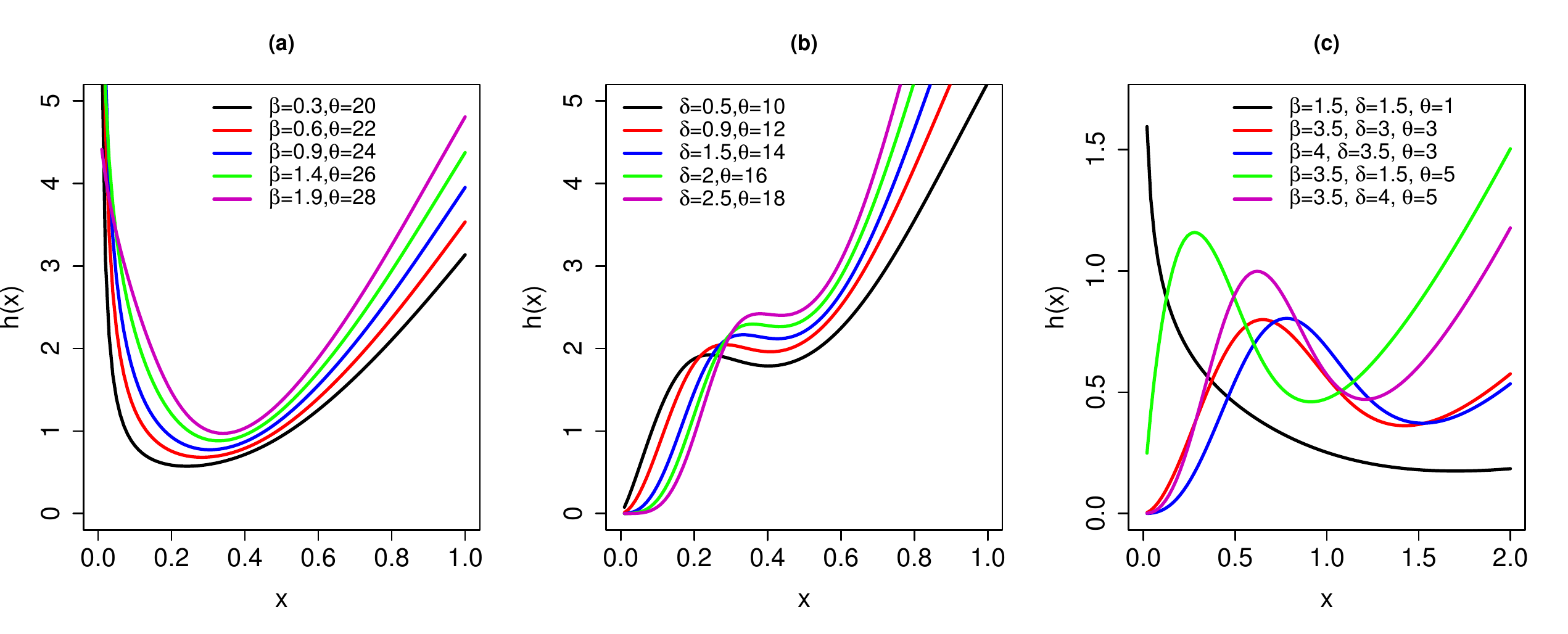}
	\caption{The GOLLGR hrf. (a) $\alpha=0.1$ and $\delta=1.5$. (b) $\alpha=0.3$ and $\beta=2.5$. (c) $\alpha=0.1$.}
	\label{gollgrfrf}

\end{figure}

\subsection{Asymptotes and quantile function}\label{sec:asymptotes}

The asymptotes below follow from Equations (\ref{cdfF}) and (\ref{pdfF})
$$
 F(x)
 \sim G(x)^{\alpha \beta},
 \quad
 f(x)\, , \, h(x)
 \sim  \alpha\beta g(x) G(x)^{\alpha \beta -1}
 \quad \text{as} \quad
 x\to 0{\color{red}^+},
 $$
 and
 $$
f(x)
\sim
\alpha\beta \, g(x) \,  \big[1-G(x)^{\beta}\big]^{\alpha-1},
\quad
h(x)
\sim
\frac{ \alpha\beta g(x)}{\left[1-G(x)^{\beta}\right]\,}
\quad \text{as} \quad
x\to \infty.
$$
For the GOLLGR distribution, we obtain
\begin{align*}
F(x)&\sim \gamma_1(\delta+1, \theta x^2)^{\alpha \beta} \, , \\
    f(x), h(x)&\sim \frac{2 \alpha\beta \theta^{\delta+1}}{\Gamma(\delta+1)}
    \,x^{2 \delta+1}\,{\rm{e}}^{- \theta x^2}\,\gamma_1(\delta+1,\theta x^2)^{\alpha\beta-1}
    \quad \text{as} \quad
    x\to 0{\color{red}^+},
\end{align*}
and
\begin{align*}
f(x)&\sim
\frac{2 \alpha\beta \theta^{\delta+1}\,x^{2 \delta+1}\,{\rm{e}}^{- \theta x^2}\big[1-\gamma_1\left(\delta+1,\theta x^2\right)^{\beta}\big]^{\alpha-1}}{\Gamma(\delta+1)},   \\
h(x)&\sim
\frac{2 \alpha\beta \theta^{\delta+1}
\,x^{2 \delta+1}\,{\rm{e}}^{- \theta x^2}}{\Gamma(\delta+1)\,\big[1-\gamma_1\left(\delta+1,
\theta x^2\right)^{\beta}\big]}
\quad \text{as} \quad
x\to \infty.
\end{align*}
%}

By combining the inverse functions of (\ref{cdfGR}) and  (\ref{cdfF}), the qf of $X$ can be expressed as
\begin{equation}\label{quantile}
x=Q(u)
=
Q_{\text{GR}}
\left(
\left[\frac{(\frac{u}{1-u})^{\frac{1}{\alpha}}}{1+(\frac{u}{1-u})^{\frac{1}{\alpha}}}\right]^{{1}/{\beta}}\,
\right),
\end{equation}
where the GR qf $Q_{\text{GR}}(z)=G_{\text{GR}}^{-1}(z;\delta,\theta)$ comes as
\begin{align}\label{Q-inverse}
Q_{\text{GR}}(z)=\big[\theta^{-1}\,\gamma^{-1}(\delta+1;z)\big]^{1/2}.
\end{align}
Here, $\gamma^{-1}\left(\delta+1;z\right)$ is the gamma
qf with shape $\delta+1$ and unity scale. So, the GOLLGR variates follow easily from \eqref{quantile}.

\subsection{Stochastic representation}\label{Stochastic representation}

The pdf of a log-logistic random variable $Y \sim \mbox{LL}(\nu,\alpha)$ is
\begin{align}\label{pdf-ll}
f_Y(y;\nu,\alpha)=
{(\alpha/\nu)(y/\nu)^{\alpha-1}\over \big[1+(y/\nu)^{\alpha}\big]^2},
\end{align}
where $\nu >0$ and $\alpha> 0$ are scale
and shape, respectively.

If $X\sim \mbox{GOLLGR}(\alpha,\beta,\delta,\theta)$ and $Y \sim \mbox{LL}(1,\alpha)$, we can write from \eqref{cdfF} and \eqref{Q-inverse}
\begin{align}
F(x;\alpha,\beta,\delta,\theta)
&=
\mathbbm{P}\biggl(
Y
\leqslant
\frac{G_{\text{GR}}(x;\delta,\theta)^{\beta}}
{1-G_{\text{GR}}(x;\delta,\theta)^{\beta}}
\biggr)\nonumber
\\[0,2cm]
&=
\mathbbm{P}\Biggl(
\biggl({Y\over 1 + Y}\biggr)^{1/\beta}
\leqslant
G_{\text{GR}}(x;\delta,\theta)
\Biggr)
%\nonumber
%\\[0,2cm]
%&
=
\mathbbm{P}\Biggl(
\Biggl[\theta^{-1}\,\gamma^{-1}\Biggl(\delta+1;\biggl(
{Y\over 1 + Y}
\biggr)^{1/\beta}\,\Biggr)\Biggr]^{1/2}
\leqslant
x
\Biggr), \label{cdf-Y}
\end{align}
where
$G_{\text{GR}}(x;\delta,\theta)=\gamma_1(\delta+1,\theta x^2)$, $x>0$.
So, we have proved the following:
\begin{proposition}
The {\rm GOLLGR} random variable $X$ admits the stochastic representation:
\begin{align*}
X=\Biggl[\theta^{-1}\,\gamma^{-1}\Biggl(\delta+1;\biggl(
{Y\over 1 + Y}
\biggr)^{1/\beta}\,\Biggr)\Biggr]^{1/2} \quad \text{for} \quad Y \sim \mbox{LL}(1,\alpha).
\end{align*}
\end{proposition}

\subsection{Critical points and modality}\label{Equation of critical points}
For brevity, we define
\begin{align}\label{def-T}
	T(x)
	=
	T(x;\beta,\delta,\theta)
	=
	\frac{G_{\text{GR}}(x;\delta,\theta)^{\beta}}
	{1-G_{\text{GR}}(x;\delta,\theta)^{\beta}}.
\end{align}

Since $T(x/k;\beta,\delta,\theta)=T(x;\beta,\delta,\theta/k^2)$, $k>0$, the next result follows from \eqref{cdf-Y}:
\begin{proposition}
If $X \sim {\rm GOLLGR}(\alpha,\beta,\delta,\theta)$, then $k\,X \sim {\rm GOLLGR}(\alpha,\beta,\delta,\theta/k^2)$.
\end{proposition}

By differentiating \eqref{cdf-Y} with respect to $x$, we obtain
\begin{align}\label{def-f-Y}
f(x;\alpha,\beta,\delta,\theta)=f_Y(T(x);1,\alpha)T'(x), \quad Y \sim \mbox{LL}(1,\alpha).	
\end{align}
Then, the derivative of $f(x;\alpha,\beta,\delta,\theta)$ is
\begin{align}\label{idet-1}
f'(x;\alpha,\beta,\delta,\theta)
=
f'_Y(T(x);1,\alpha)[T'(x)]^2+f_Y(T(x);1,\alpha)T''(x).
\end{align}
Since
\begin{align*}
f'_Y(t;1,\alpha)
=
-f_Y(t;1,\alpha)r[t] \quad
\text{with} \
r[t]={t^\alpha+\alpha(t^\alpha-1)+1\over t(t^\alpha+1)},
\end{align*}
Equation \eqref{idet-1} can be expressed as
\begin{align*}
f'(x;\alpha,\beta,\delta,\theta)
=
f_Y(T(x);1,\alpha)\big\{T''(x)-r[T(x)][T'(x)]^2\big\},
\end{align*}
where
\begin{align*}
T'(x)
&=
{\beta g_{\text{GR}}(x;\delta,\theta) T(x)\over G_{\text{GR}}(x;\delta,\theta)[1-G_{\text{GR}}(x;\delta,\theta)^\beta]},
\\[0,2cm]
	T''(x)
	&=
	T'(x)
	\biggl\{
	{g'_{\text{GR}}(x;\delta,\theta)\over g_{\text{GR}}(x;\delta,\theta)}
	+
	g_{\text{GR}}(x;\delta,\theta)
	{(\beta+1)G_{\text{GR}}(x;\delta,\theta)^\beta+\beta-1\over G_{\text{GR}}(x;\delta,\theta) [1-G_{\text{GR}}(x;\delta,\theta)^\beta]}
	\biggr\},
\end{align*}
and
\begin{align}\label{def-der-g}
{g'_{\text{GR}}(x;\delta,\theta)\over g_{\text{GR}}(x;\delta,\theta)}
=
{2\delta+2\theta x^2-1\over  x}.
\end{align}
Then,
\begin{eqnarray*}
	&f'(x;\alpha,\beta,\delta,\theta)=
%	\\[-0.15cm]
%	&=
	\\[0.1cm]%[-0.15cm]
	&f_Y(T(x);1,\alpha)T'(x)
	\biggl\{
	{g'_{\text{GR}}(x;\delta,\theta)\over g_{\text{GR}}(x;\delta,\theta)}
	+
	{g_{\text{GR}}(x;\delta,\theta)\over G_{\text{GR}}(x;\delta,\theta)}\, [T(x)+1]
	\Big\{
	{(\beta+1)G_{\text{GR}}(x;\delta,\theta)^\beta}
	-
	\alpha\beta
	\Big[
	{T(x)^\alpha-1\over T(x)^\alpha+1}
	\Big]
	-1
	\Big\}
	\biggr\}.
\end{eqnarray*}
Equation \eqref{def-f-Y} gives   $f_Y(T(x);1,\alpha)T'(x)=f(x;\alpha,\beta,\delta,\theta)$, which is a positive function. Hence, any critical point of the pdf of $X$ satisfies the non-linear equation:
\begin{align*}
		{g'_{\text{GR}}(x;\delta,\theta)\over g_{\text{GR}}(x;\delta,\theta)}
	+
	{g_{\text{GR}}(x;\delta,\theta)\over G_{\text{GR}}(x;\delta,\theta)}\, [T(x)+1]
	\Biggl\{
	{(\beta+1)G_{\text{GR}}(x;\delta,\theta)^\beta}
	-
	\alpha\beta
	\biggl[
	{T(x)^\alpha-1\over T(x)^\alpha+1}
	\biggr]
	-1
	\Biggr\}=0.
\end{align*}
The previous result can be written equivalently as:
\begin{proposition}
A critical point of the pdf of $X$ satisfies
\begin{align}\label{eq-crit-points}
%	{(2\delta-1)+2\theta x^2\over x}
	{y''\over (y')^2}
	+
	{
	{(\beta+1)y^\beta}[y^{\alpha\beta}+(1-y^\beta)^\alpha ]
	-
	(\alpha\beta+1)
	{y^{\alpha\beta}-2(1-y^\beta)^\alpha}
	\over
	y(1-y^\beta)[y^{\alpha\beta}+(1-y^\beta)^\alpha ]}
=0,
\end{align}
where $y=y(x)=G_{\text{GR}}(x;\delta,\theta)=\gamma_1(\delta+1,\theta x^2)$.
\end{proposition}

%When $\delta=1/2$ in \eqref{def-der-g}; $y''/ y'=2\theta x$. In addition, for $\alpha=1$, the above equation is
%\begin{align*}
%{y'\over y}=\theta x.
%\end{align*}
%Equivalently,
%\begin{align*}
%	y=G_{\text{GR}}(x;\delta,\theta)={\rm e}^{\theta x^2/ 2}.
%\end{align*}
%Since $0<\gamma_1(\delta+1,z)<1$ and ${\rm e}^{z/ 2}>1$ on $(0,\infty)$, it is clear that the equation $\gamma_1(\delta+1,z)={\rm e}^{z/ 2}$ has no roots in the interval $(0,\infty)$. In other words, the GOLLGR pdf has no critical point on $(0,\infty)$. Now, if $0<\beta< 1/3$, then, by the sixth limit in \eqref{limit-f}, $\lim_{x\to 0^+}f(x;\alpha,\beta,\delta,\theta)=\infty$. But, $\lim_{x\to \infty}f(x;\alpha,\beta,\delta,\theta)=0$. This implies that the  GOLLGR pdf is decreasing on $(0,\infty)$.

\bigskip
In the remainder of this section, we use Equation \eqref{eq-crit-points} and the limit in \eqref{limit-f}
to analyze the modality of the pdf of $X$ when $\alpha=1$.

For $\alpha=1$, Equation \eqref{eq-crit-points} reduces to
\begin{align*}
{g_{\text{GR}}(x;\delta,\theta)\over G_{\text{GR}}(x;\delta,\theta)}
=
{\delta+\theta x^2-(1/2)\over  x}.
\end{align*}
Equivalently,
\begin{align}\label{eq-G-g}
G_{\text{GR}}(x;\delta,\theta)
=
%{x g_{\text{GR}}(x;\delta,\theta) \over  \delta+\theta x^2-(1/2)}
%=
%\frac{2 \theta^{\delta+1}}{\Gamma(\delta+1)}\,{ x^{2 (\delta+1)}\,{\rm{e}}^{- \theta x^2}\over \delta+\theta x^2-(1/2)}\eqqcolon
g(x),
\end{align}
where
\begin{align*}
g(x)
=
x\,g_{\text{GR}}(x;\delta,\theta)
=
\frac{2 \theta^{\delta+1}}{\Gamma(\delta+1)}\,{ x^{2 (\delta+1)}\,{\rm{e}}^{- \theta x^2}\over \delta+\theta x^2-(1/2)}.
\end{align*}
A careful analysis shows that, on $(0,\infty)$, $g(x)$ is decreasing/decreasing-increasing-decreasing when $\delta<1/2$ or  unimodal when $\delta\geqslant 1/2$, and $\lim_{x\to 0^+}g(x)=\lim_{x\to\infty}g(x)=0$. Moreover, $g(x)$ has a
vertical asymptotic at $x_{\rm va} = \sqrt{[(1/2)-\delta]/\theta}$ when $\delta<1/2$, with $g(x)<0$ for all $x<x_{\rm va}$ and $g(x)>0$ for all $x>x_{\rm va}$.
Since $G_{\text{GR}}(x;\delta,\theta)$ is increasing on $(0,\infty)$, because this one is a cdf, it is plausible
to expect that, by varying the parameters, Equation \eqref{eq-G-g} has at most three roots on $(0,\infty)$. So, the pdf of $X$ has at
most three critical points on $(0,\infty)$. In the following,
we analyze some possible scenarios:

\begin{itemize}
%\item
%Let $\delta=-1$. In this case, $g$ is decreasing and has a vertical asymptote at $x_{\rm va} = \sqrt{3/(2\theta)}$, and $g(x)>0$ for all $x>x_{\rm va}$. Then, it is clear that  Equation \eqref{eq-G-g} has a single root, denoted by $x_0$.

%If Equation \eqref{eq-G-g} has no roots and $\delta=-1$,  by following limit, obtained of \eqref{limit-f} for $\beta\to\infty$,
%\begin{align}\label{limit-f-1}
%\lim_{x\to 0^+} f(x;\alpha,\beta,\delta,\theta)
%=
%0, \quad \delta\neq -1/2, \delta>-1,
%\end{align}
%the  GOLLGR pdf is decreasing.

\item[$\bullet$] If the GOLLGR pdf has a single critical point,
say $x_0$, and $\beta>1$ and $\delta\geqslant 1/2$,
we have by the second limit in \eqref{limit-f}
$
\lim_{x\to 0^+} f(x;\alpha,\beta,\delta,\theta)
=
\lim_{x\to\infty} f(x;\alpha,\beta,\delta,\theta)
=
0
$.
Then, $x_0$ is really a maximum point, and the pdf of $X$
is unimodal.

%\item If Equation \eqref{eq-G-g} has two roots, denoted by $x_1<x_2$, and $\delta= -1$,  by the sixth limit in \eqref{limit-f},  $\lim_{x\to 0^+}f(x;\alpha,\beta,\delta,\theta)=\infty$ and $\lim_{x\to \infty}f(x;\alpha,\beta,\delta,\theta)=0$. Then, $x_1$ is a minimum point and  $x_2$ is a maximum point.
%Hence, it follows that the GOLLGR pdf is decreasing-increasing-decreasing.

\item[$\bullet$] If the GOLLGR pdf has three single critical points, say $x_1<x_2<x_3$, and $\beta>1$
and $\delta\geqslant 1/2$, again, by the second limit in \eqref{limit-f}, we have
$
\lim_{x\to 0^+} f(x;\alpha,\beta,\delta,\theta)
=
\lim_{x\to\infty} f(x;\alpha,\beta,\delta,\theta)
=0$.
Hence, $x_1$ and $x_3$ are maximum points, and $x_2$ is
a minimum point, and the GOLLGR pdf is bimodal.
\end{itemize}

In general, as done previously, we can use the limit in \eqref{limit-f} and the number of critical points of the
pdf of $X$ to obtain the result:
\begin{theorem}
The {\rm GOLLGR} pdf is decreasing/ decreasing-increasing-decreasing/unimodal
or bimodal.
\end{theorem}

\subsection{Tail behavior}

Here, we prove that, under certain constraints in the parameter space, the distribution of $X$ has thinner tails than an exponential distribution. More precisely, we prove the following two results:
\begin{proposition}
%For $\alpha\geqslant 1$,  the {\rm GOLLGR} has upper light-tail distribution. In other words,
For any $\alpha\geqslant 1$ and any $t>0$,
\begin{align}\label{lim-object}
\lim_{x\to \infty}
{
	{\rm e}^{-tx}
	\over
	1-F(x;\alpha,\beta,\delta,\theta)
}
=\infty.
\end{align}
\end{proposition}
\begin{proof}
If $Y \sim \mbox{LL}(1,\alpha)$, by \eqref{cdf-Y}, $F(x;\alpha,\beta,\delta,\theta)=
\mathbbm{P}(Y\leqslant T(x))$. Moreover, it is well-known that $\mathbbm{P}(Y\leqslant y)=y^\alpha/(1+y^\alpha)$.
Then, from the definition \eqref{def-T} of $T$, we have
(for any $t>0$),
\begin{align}\label{ineq-in}
{
{\rm e}^{-tx}
\over
1-F(x;\alpha,\beta,\delta,\theta)
}
=
{
	{\rm e}^{-tx}
	\over
	1-\mathbbm{P}(Y\leqslant T(x))
}
&=
{\rm e}^{-tx}[1+T(x)^\alpha]
\nonumber
\\[0,2cm]
&\geqslant
{\rm e}^{-tx}
T(x)^\alpha
=
{
{\rm e}^{-tx} \,
G_{\text{GR}}(x;\delta,\theta)^{\alpha\beta}
\over
\big[1-G_{\text{GR}}(x;\delta,\theta)^{\beta}\big]^\alpha
}.
\end{align}

The L'Hôpital's rule yields
\begin{align}\label{hospital}
\lim_{x\to \infty}
{
	{\rm e}^{-tx}
	\over
	\big[1-G_{\text{GR}}(x;\delta,\theta)^{\beta}\big]^\alpha
}
=
\lim_{x\to \infty}
{
	t \big[{{\rm e}^{-tx}\over	g_{\text{GR}}(x;\delta,\theta)}\big]
	\over
	\alpha\beta\big[1-G_{\text{GR}}(x;\delta,\theta)^{\beta}\big]^{\alpha-1} G_{\text{GR}}(x;\delta,\theta)^{\beta-1}
}.
\end{align}
Since (for $\alpha\geqslant 1$),
\begin{align*}
\lim_{x\to \infty}
{
	{\rm e}^{-tx}
	\over
	g_{\text{GR}}(x;\delta,\theta)	
}
=
\biggl[
\frac{2 \theta^{\delta+1}}{\Gamma(\delta+1)}\,x^{2 \delta+1}\,{\rm{e}}^{- \theta x^2+tx}
\biggr]^{-1}
=
\infty
\end{align*}
and $\lim_{x\to \infty}G_{\text{GR}}(x;\delta,\theta)=1$.
We have from \eqref{hospital}
\begin{align*}
\lim_{x\to \infty}
{
	{\rm e}^{-tx}
	\over
	\big[1-G_{\text{GR}}(x;\delta,\theta)^{\beta}\big]^\alpha
}
=
\infty.
\end{align*}
By taking $x\to\infty$ for both sides of inequality \eqref{ineq-in} and by using the above limit, it follows \eqref{lim-object}.
\end{proof}

\begin{proposition}
For any $0<\beta\leqslant 1$ and $\delta>0$, the limit (for $t>0$) \eqref{lim-object} holds.
\end{proposition}
\begin{proof}
By inequality in \eqref{ineq-in}, it is enough to prove
\begin{align}\label{claim}
\lim_{x\to \infty}
{
	{\rm e}^{-tx}
	\over
	\big[1-G_{\text{GR}}(x;\delta,\theta)^{\beta}\big]^\alpha
}
=
\lim_{x\to \infty}
{
	1
	\over
	{\rm e}^{tx}\big[1-\gamma_1(\delta+1,\theta x^2)^{\beta}\, \big]^\alpha
}
=\infty.
\end{align}

Indeed, by using the $C_p$ inequality:
\begin{align*}
	\forall x,y\geqslant 0; \ (x+y)^p\leqslant C_p(x^p+y^p), \quad
	\text{where} \ p>0 \ \text{and} \ C_p=\max\{1,2^{p-1}\};
\end{align*}
we have (for $0<\beta\leqslant 1$)
\begin{align}\label{ineq-1}
	1-\gamma_1(\delta+1,\theta x^2)^{\beta}
	\leqslant
	[1-\gamma_1(\delta+1,\theta x^2)]^{\beta}
	=
	\Gamma_1(\delta+1,\theta x^2)^{\beta},
\end{align}
where $\Gamma_1(p,x)=\Gamma(p)^{-1}\int_{x}^{\infty}w^{p-1}\,
{\rm{e}}^{-w} d w$ is the upper incomplete gamma function ratio.

By using the inequality of Natalini and Palumbo (2000): for $a > 1$, $B > 1$ and $x >B(a-1)/(B-1)$,
\begin{align*}
\Gamma(a,x)<B\, x^{a-1}\, {\rm e}^{-x};
\end{align*}
we have (for $ x >\sqrt{B\delta/[\theta(B-1)]}$ and $\delta> 0$)
\begin{align}\label{ineq-2}
\Gamma_1(\delta+1,\theta x^2)^{\beta}<
	B^\beta\, \theta^{\beta\delta}\, \Gamma(\delta+1)^{-\beta} x^{2\beta\delta}\, {\rm e}^{- \beta\theta x^2}.
\end{align}

By combining \eqref{ineq-1} and \eqref{ineq-2}, we obtain (for any $ x >\sqrt{B\delta/[\theta(B-1)]}$)
\begin{align*}
		{\rm e}^{tx}\big[1-\gamma_1(\delta+1,\theta x^2)^{\beta}\, \big]^\alpha
		<
		B^{\alpha\beta} \theta^{\alpha\beta\delta} \Gamma(\delta+1)^{-\alpha\beta} x^{2\alpha\beta\delta} {\rm e}^{-\alpha\beta\theta x^2+tx}.
\end{align*}
Letting $x\to \infty$ in the above inequality, we have ${\rm e}^{tx}\big[1-\gamma_1(\delta+1,\theta x^2)^{\beta}\, \big]^\alpha$ tends to zero, proving the limit in \eqref{claim}. Thus, we complete the proof.
\end{proof}

\subsection{Linear Representation}\label{sec:linear_rep}

First, the {\it exponentiated-G} (``Exp-G'') random variable
$W\sim\mbox{Exp}^c$G for a continuous cdf $G(x)$ and $c>0$,
has cdf $H_c(x)=G(x)^c$ and pdf $h_c(x)=c\,g(x)\,G(x)^{c-1}$.
Many Exp-G properties were published in the last three decades.

We obtain after some algebra the power series
\begin{equation}\label{denF}
G(x)^{\alpha\beta}+\big[1-G(x)^\beta\big]^\alpha
=\sum_{k=0}^{\infty} c_k\, G(x)^k \, ,
\end{equation}
where $a_k=a_k(\alpha \beta)
=\sum_{j=k}^\infty  (-1)^{j+k}
    \binom{\alpha \beta}{j} \binom{j}{k}$, and
$$c_k=c_k(\alpha,\beta)=a_k(\alpha \beta)+
\sum_{i=0}^\infty \sum_{j=k}^\infty
(-1)^{i+j+k} \binom{\alpha}{i} \binom{i\beta}{j} \binom{j}{k}$$.

Equation (\ref{cdfF}) can be rewritten from the ratio of two power series as
\begin{eqnarray}\label{expcdf}
F(x)=\sum_{k=0}^\infty d_k\,G(x)^k,
\end{eqnarray}
where $d_k=d_k(\alpha,\beta)$'s
are found found recursively (for $k > 0$, $d_0=a_0/c_0$)
$$d_k=c_0^{-1}\left(a_k+
\sum_{r=1}^k c_r\,d_{k-r}\right).$$

By differentiating (\ref{expcdf}) and changing indices
\begin{equation*}
f(x)=\sum_{l=0}^\infty
(l+1)\,d_{l+1}\,g(x)\, G(x)^{l} \, .
\end{equation*}
For the GR model, we obtain
\begin{equation}\label{linereppdfGOLLG}
f(x)=\sum_{l=0}^\infty
(l+1)\, d_{l+1}(\alpha,\beta)\,
\frac{2 \theta^{\delta+1}}{\Gamma(\delta+1)}
\,x^{2 \delta+1}\,{\rm{e}}^{- \theta x^2}
\, \gamma_1(\delta+1,\theta x^2)^{l} ,
\end{equation}
The power series for the incomplete gamma function ratio
holds
\begin{equation*}\label{expdistfun}
\gamma_1(\delta+1,\theta x^2)
=\frac{(\theta x^2)^{\delta+1}}{\Gamma(\delta+1)}
\sum_{m=0}^\infty \frac{(-1)^m
(\theta x^2)^m}{m!\, (\alpha+1+m)}\, \cdot
\end{equation*}

Equation 0.314 in Gradshteyn and Ryzhik (2000) gives
(for a natural number $l\ge 1$)
$$
\left(\sum_{m=0}^\infty q_m \,x^m\right)^l
=\sum_{m=0}^\infty e_{m}^{(l)}\, x^m \, ,
$$
where $e_{0}^{(l)}=q_0^l$, and $e_{m}^{(l)}$ (for $l\ge 1$)
can be found from
\begin{equation}\label{recurrent}
e_{m}^{(l)}=\frac{1}{m\, q_0}\sum_{i=1}^{m}[(l+1)\,i-m]\, q_i \, e_{m-i}^{(l)}.
\end{equation}

Then,
\begin{equation}\label{power}
\gamma_1(\delta+1,\theta x^2)^l
=\frac{(\theta x^2)^{l(\delta+1)}}{\Gamma(\delta+1)^l}
\sum_{m=0}^\infty e_{m}^{(l)} \,  x^{2m}\, ,
\end{equation}
where the quantities $e_{m}^{(l)}$ follow
from (\ref{recurrent}) with the constants
$$
q_m=\frac{(-1)^m
\theta^m}{(\delta+1+m)m!}
$$
for $m=0,1,\ldots$.

Further, we set the conditions $e_0^{(0)}=1$ and
$e_m^{(0)}=0$ for $m\ge 1$.
Hence, inserting (\ref{power}) in Equation (\ref{linereppdfGOLLG}) (under these conditions) yields
 \begin{align*}
f(x)
&=\sum_{l,m=0}^\infty
     \frac{2\,(l+1)
     \, d_{l+1}(\alpha,\beta)
     \, e_{m}^{(l)}}{\Gamma(\delta+1)^{l+1}\theta^m}
     \,\theta^{[l(\delta+1)+m+\delta]+1}x^{2[l(\delta+1)+m+\delta]+1}
     {\rm{e}}^{-\theta x^2}
\end{align*}
and then
\begin{align}\label{combination}
    f(x)
    &=
     \sum_{l,m=0}^\infty
     w_{l,m}\,
     g_{\text{GR}}(x;\theta,\delta_{l,m}^\ast) \, ,
\end{align}
where $\delta_{l,m}^*=l(\delta+1)+m+\delta$, and
the coefficients are
\begin{align*}
    w_{l,m}=w_{l,m}(\theta,\delta,\alpha,\beta)
    =
     \frac{(l+1)\,\Gamma(\delta_{l,m}^\ast+1)
     \, d_{l+1}(\alpha,\beta)
     \, e_{m}^{(l)}}{\theta^m\,\Gamma(\delta+1)^{l+1}} \, \cdot
\end{align*}

Equation (\ref{combination}) is useful to obtain some properties
of $X$ from those of the GR model.

%\pagebreak

\subsection{Properties}\label{sec:properties}

The $s$th ordinary moment of $X$ comes from (\ref{momentGR}) and (\ref{combination}) as
\begin{equation*}\label{moment}
E(X^s)=
\sum_{l,m=0}^\infty
w_{l,m}(\theta,\delta,\alpha,\beta)
\frac{\Gamma(\nicefrac{s}{2}+\delta_{l,m}^\ast+1)}{\theta^{s/2}\,
\Gamma(\delta_{l,m}^\ast+1)}\, \cdot
\end{equation*}

The $s$th incomplete moment of $X$ follows from (\ref{combination}) as
\begin{align*}
    m_s(x)
     &=
     \sum_{l,m=0}^\infty
     \frac{w_{l,m}(\theta,\delta,\alpha,\beta)}{
     \Gamma(\delta_{l,m}^\ast+1)}
     \int_0^x 2\theta^{\delta_{l,m}^\ast+1}\, t^s\, t^{2\delta_{l,m}^\ast+1}\,
     {\rm e}^{-\theta t^2} \, dt  \\
     &=
     \sum_{l,m=0}^\infty
     w_{l,m}(\theta,\delta,\alpha,\beta)
     \frac{\Gamma(\delta_{l,m}^\ast+\nicefrac{s}{2}+1)
     }{\Gamma(\delta_{l,m}^\ast+1) \theta^{\nicefrac{s}{2}}}
     \int_0^x \frac{2\theta^{\delta_{l,m}^\ast+\nicefrac{s}{2}+1}}{
     \Gamma(\delta_{l,m}^\ast+\nicefrac{s}{2}+1)}
     \, t^{2(\delta_{l,m}^\ast+\nicefrac{s}{2})+1}
     {\rm e}^{-\theta t^2} \, dt  \\
     &=
     \sum_{l,m=0}^\infty
     w_{l,m}(\theta,\delta,\alpha,\beta)
     \frac{\Gamma(\nicefrac{s}{2}+\delta_{l,m}^\ast+1)
     }{\theta^{\nicefrac{s}{2}} \, \Gamma(\delta_{l,m}^\ast+1)}
     \gamma_1(\delta_{l,m}^\ast+\nicefrac{s}{2}+1,\theta x^2)
      \, .
\end{align*}

The mean deviations and Bonferroni and Lorenz curves of $X$ are
obtained from $m_1(x)$.

The generating function (gf) of $X$ can be expressed as
\begin{align*}
M(t)
=
\int_0^\infty
{\rm{e}}^{tx}\,
\sum_{l,m=0}^\infty
w_{l,m}(\theta,\delta,\alpha,\beta)\,
g_{\text{GR}}(x; \theta,\delta_{l,m}^\ast) \, dx \, ,
\end{align*}
that is,
%\begin{align*}
%M(t)
%&=
%\sum_{l,m=0}^\infty
%\frac{w_{l,m}(\theta,\delta,\alpha,\beta)
%\,2\theta^{[l(\delta+1)+m+\delta]+1}}{\Gamma([l(\delta+1)+m+
%\delta]+1)}
%\int_0^\infty x^{2[l(\delta+1)+m+\delta]+1}\,
%{\rm e}^{-\theta x^2+tx}\,dx \, .
%\end{align*}
\begin{align*}
M(t)
&=
\sum_{l,m=0}^\infty
\frac{2\, w_{l,m}(\theta,\delta,\alpha,\beta)
\, \theta^{\delta_{l,m}^\ast+1}}{\Gamma(\delta_{l,m}^\ast+1)}
\int_0^\infty x^{2\delta_{l,m}^\ast+1}\,
{\rm e}^{-\theta x^2+tx}\,dx \, .
\end{align*}

From Equation 2.3.15.3 in Prudnikov {\it et al.} (1986), we can write
  $$
  \int_0^\infty x^{\alpha-1}{\rm e}^{-px^2-qx} dx=
  \frac{\Gamma(\alpha)}{(2p)^{\nicefrac{\alpha}{2}}}
  \exp \left( \frac{q^2}{8p}\right)\,
  D_{-\alpha}\left( \frac{q}{\sqrt{2p}}\right) \, ,
  $$
where  $\text{Re}(\alpha), \,\text{Re}(p) >0$, and
\begin{equation*}
D_p(y)= \frac{\exp(-y^2/4)}{\Gamma(-p)}\int_0^{\infty}\exp\{-(wy+w^2/2)\}w^{-(p+1)}\,dw.
\end{equation*}
Thus,
\begin{align*}
M(t)
&=
\sum_{l,m=0}^\infty
\frac{2\, w_{l,m}%(\theta,\delta,\alpha,\beta)
\, \theta^{\delta_{l,m}^\ast+1}}{\Gamma(\delta_{l,m}^\ast+1)}
\frac{\Gamma(\tilde{\delta}_{l,m})}{(2\theta)^{\nicefrac{\tilde{\delta}_{l,m}}{2}}}
\exp\left(\frac{t^2}{8\theta}\right)D_{-\tilde{\delta}_{l,m}}\left(- \frac{t}{\sqrt{2\theta}}\right),
\end{align*}
where $\tilde{\delta}_{l,m}=2\,(\delta_{l,m}^\ast+1)$.

\subsection{Estimation}\label{sec:estimation}

We obtain the maximum likelihood estimate (MLE) of
$\etn=(\alpha,\beta,\delta,\theta)^\top$ given the data  $x_1,\dots,x_n$ from the GOLLGR distribution.

The total log-likelihood function for $\etn$ is
\begin{align}\label{lntau}
l_n(\etn)
&=
n\left[ \log (\alpha \beta) %+\log2
+(\delta+1)\log \theta \right]
+(2\delta+1) \sum_{i=1}^n \log x_i
-\theta
\sum_{i=1}^n x_i^2 \nonumber \\
%&\sum_{i=1}^n \log \left(\frac{2 \theta^{\delta+1}}{\Gamma(\delta+1)}
%\,x_i^{2 \delta+1}\,{\rm{e}}^{- \theta x_i^2}\right) \\
&+(\alpha\beta-1) \sum_{i=1}^n \log
\gamma_1(\delta+1,\theta x_i^2)
+(\alpha -1)
\sum_{i=1}^n\log[1-
\gamma_1(\delta+1,\theta x_i^2)^\beta] \nonumber \\
&-
2 \sum_{i=1}^n \log \big\{ \gamma_1(\delta+1,\theta x_i^2)^{\alpha\beta}
+[1-\gamma_1(\delta+1,\theta x_i^2)^\beta]^\alpha\big\}.
\end{align}

The maximization of (\ref{lntau}) can be done
 using the {\tt R} software or SAS (PROC NLMIXED),
 among others.

\section{The GOLLGR regression}\label{sec:regression}

Recently some papers on regression models have been published, for example, see, Hashimoto {\it et al.} (2019), Prataviera {\it et al.} (2020), 
Silva {\it et al.} (2020) and Vasconcelos {\it et al.} (2021). Based on these papers we introduced the regression model based on the GOLLGR distribution.

The systematic components of the GOLLGR re\-gression for $X$ are defined by
 \begin{eqnarray}\label{reg_1}
\theta_{i}=\exp(\vn_{i}^{\top}\lamn_{1}) \qquad \text{and} \qquad \delta_{i}= \exp (\vn_{i}^{\top}\lamn_{2})-1 \qquad i=1,\ldots,n,
\end{eqnarray}
where $\lamn_1=(\lambda_{11},\ldots,\lambda_{1p})^{\top}$  and $\lamn_2=(\lambda_{21},\ldots,\lambda_{2p})^{\top}$ are vectors of unknown coefficients, and $\vn_{i}^{\top}=(v_{i1},\ldots,v_{ip})$ is a vector of known explanatory variables.

The survival function of $X$ comes from (\ref{cdfF}) as
\begin{eqnarray}\label{reg_gollfw}
S(x_i|\vn)
=
\frac{[1-\gamma_{1}(\delta_i+1,\theta_i\,x_{i}^{2})^{\beta}]^{\alpha}}{\gamma_{1}(\delta_i+1,\theta_i\,x_{i}^{2})^{\alpha\,\beta}
+\left[1-\gamma_{1}(\delta_i+1,\theta_i\,x_{i}^{2})^{\beta}\right]^{\alpha}}.
\end{eqnarray}

Equation (\ref{reg_gollfw}) opens new possibilities for fitting different types of regressions. The odd log-logistic GR (OLLGR) regression follows when $\beta=1$, the exponentiated GR (EGR) regression when $\alpha=1$, and the GR regression when $\beta=\alpha=1$.

Let $X_i$ be the lifetime and $C_i$ be the non-informative censoring time (assuming independent), and $x_i=\min\{X_i,C_i\}$. The total log-likelihood function for
$\etn=(a,b,\lamn_{1}^{\top})^{\top}$ from regression (\ref{reg_1}) is
\begin{align}\label{veroy}
l(\etn)&=
r \log\left(\frac{\alpha\,\beta\,2}{\Gamma(\delta_i+1)}
\right)
+(\delta_i+1)
\sum_{i \in F}\log(\theta_i)
+(2\,\delta_i+1)
\sum_{i \in F}\log(x_i)
-\sum_{i \in F}\theta_{i}x_{i}^{2}\nonumber\\
&+(\alpha\,\beta-1)
\sum_{i \in F}\log[\gamma_{1}(\delta_i+1,\theta_i\,x_{i}^{2})]
+(\alpha-1)
\sum_{i \in F}\log[1
-\gamma_{1}(\delta_i+1,\theta_i\,x_{i}^{2})^\beta]
\nonumber\\
&-2\sum_{i \in F}\log\left\{\gamma_{1}(\delta_i+1,\theta_i\,x_{i}^{2})^{\alpha\,\beta}
+[1-\gamma_{1}(\delta_i+1,\theta_i\,x_{i}^{2})^{\beta}]^{\alpha}\right\}\nonumber\\
&+\sum_{i \in C} \log\left\{\frac{[1-\gamma_{1}(\delta_i+1,\theta_i\,x_{i}^{2})^{\beta}]^{\alpha}}{\gamma_{1}(\delta_i+1,\theta_i\,x_{i}^{2})^{\alpha\,\beta}
+\left[1-\gamma_{1}(\delta_i+1,\theta_i\,x_{i}^{2})^{\beta}\right]^{\alpha}}\right\},
\end{align}
where $F$ and $C$ are the sets of observed lifetimes and
censoring times, $r$ is the number of uncensored observations (failures).
The MLE $\widehat{\etn}$ of the vector of unknown parameters can be
found by maximizing (\ref{veroy}).

\subsection{Residual analysis}

The quantile residuals (qrs) (Dunn and Smith, 1996)
have been adopted frequently in regression
applications to verify possible deviations from the model
assumptions.
For the proposed regression, they are
\begin{eqnarray}\label{resq}
qr_{i}=\Phi^{-1}
%\begin{Bmatrix}
\left\{
\frac{\gamma_1(\hat{\delta_{i}}+1,\hat{\theta_{i}} x_{i}^2)^{\hat{\alpha} \hat{\beta}}}
{\gamma_1(\hat{\delta_{i}}+1,\hat{\theta_{i}} x_{i}^2)^{\hat{\alpha} \hat{\beta}}+\big[1-\gamma_1(\hat{\delta_{i}}+1,\hat{\theta_{i}} x_{i}^2)^{\hat{\beta}}\big]^{\hat{\alpha}}}
\right\},
%\end{Bmatrix}
\end{eqnarray}
where $\Phi^{-1}(\cdot)$ is the standard normal qf.

\subsection{Simulations}\label{sec:simulations}

In this section, two simulation studies are presented by using \texttt{gamlss} package in {\tt R} software.\\

{\bf The GOLLGR distribution}\\

First, we generate 1,000 samples from Equation (\ref{quantile}) with $\alpha=0.35$, $\beta=0.55$, $\delta=-0.55$ and $\theta=0.11$, and sample sizes $n=50$, $150$ and $500$, and
calculate the MLEs. The average estimates (AEs), biases and mean squared errors (MSEs) in Table \ref{simulation1} show that
the AEs converge to the true values and the biases and MSEs decrease when $n$ increases.\\

\begin{table}[!htb]
\centering{\caption{Simulation results from the GOLLGR distribution.\label{simulation1}} \vspace*{0.3cm}
\begin{tabular}{c|ccc|ccc|ccc}
\hline \hline
   & $n=50$ & & & $n=150$ && &  $n=500$ &  &  \\
  \hline \hline
  Parameter & AE & Bias & MSE  & AE & Bias & MSE & AE & Bias & MSE \\
  \hline \hline
 $\alpha$ & 0.3533 & 0.0033 & 0.0205 &    0.3487 & -0.0013 & 0.0050  & 0.3488 & -0.0012 & 0.0027 \\
    $\beta$ & 0.7140 & 0.1640 & 0.3857  &  0.5861 & 0.0361 & 0.0197  & 0.5688 & 0.0188 & 0.0095 \\
    $\delta$ & -0.5217 & 0.0283 & 0.1054    & -0.5512 & -0.0012 & 0.0044  & -0.5514 & -0.0014 & 0.0018 \\
    $\sigma$ & 0.1484 & 0.0384 & 0.0115   & 0.1188 & 0.0088 & 0.0015  & 0.1148 & 0.0048 & 0.0008 \\
    \hline  \hline
  \end{tabular}}
\end{table}

{\bf The GOLLGR regression}\\

The second study examines the accuracy of the MLEs in the
proposed regression. The observations are generated from $X_{i} \sim \text{GOLLGR} (\alpha,\beta,\delta_{i}, \theta_{i})$ and $v_{1i} \sim \text{Binomial} (1,0.5)$, where $v_{1i}$ is taken
in two groups (0 and 1). We consider 1,000 samples for $\alpha=0.37$, $\beta=0.61$, $\lambda_{10}=0.55$, $\lambda_{11}=1.75$, $\lambda_{20}=0.65$ and $\lambda_{21}=2.75$ and $n$ = $150$, $350$ and $650$, and the simulation process follows as:

(i) Generate $v_{1i} \sim \text{Binomial} (1,0.5)$;

(ii) Calculate $\delta_{i}$ and $\theta_{i}$ from the systematic components: $\delta_{i}=\text{exp}(\lambda_{10}+\lambda_{11}v_{1i})-1$ and $\theta_{i}= \text{exp}(\lambda_{20}+\lambda_{21}v_{1i})$, respectively;

(iii) Generate $u_{i}\sim U(0,1)$;

(iv) The previous steps yield $x_{i}$'s from (\ref{quantile}).

The numbers in Table \ref{simulation2} reveal that the AEs
converge to the true values, and the biases and MSEs decrease when $n$ increases, thus indicating the consistency of the estimators.

\begin{table}[!htb]
\centering{\caption{Simulation results from the GOLLGR regression.\label{simulation2}} \vspace*{0.3cm}
\begin{tabular}{c|ccc|ccc|ccc}
  \hline \hline
   & $n=150$ & & & $n=350$ && &  $n=650$ &  &  \\
  \hline \hline
  Parameters & AE & Bias & MSE & AE & Bias & MSE &  AE & Bias & MSE \\
  \hline \hline
  $\alpha$ & 0.3736 & 0.0036 & 0.0146 &  0.3700 & 0.0000 & 0.0055 &   0.3722 & 0.0022 & 0.0028  \\
     $\beta$ & 0.7324 & 0.1224 & 0.2991 &  0.6615 & 0.0515 & 0.0897 &  0.6317 & 0.0217 & 0.0332  \\
$\lambda_{10}$ & 0.5917 & 0.0417 & 0.2022 &  0.5787 & 0.0287 & 0.0973 &   0.5618 & 0.0118 & 0.0467 \\
$\lambda_{11}$ & 1.7715 & 0.0215 & 0.0535 &  1.7586 & 0.0086 & 0.0219 &  1.7553 & 0.0053 & 0.0114 \\
   $\lambda_{20}$ & 0.7308 & 0.0808 & 0.1444 &  0.6955 & 0.0455 & 0.0606 &   0.6640 & 0.0140 & 0.0307  \\
  $\lambda_{21}$ & 2.7400 & -0.0100 & 0.0582 &  2.7450 & -0.0050 & 0.0249 &  2.7518 & 0.0018 & 0.0127  \\
   \hline \hline
\end{tabular}}
\end{table}

\section{Applications}\label{sec:applicatons}

Here, we compare the GOLLGR model and its GR, EGR,
OLLGR sub-models in two applications. We determine the MLEs,
and the criteria: Akaike Information Criterion (AIC), Consistent Akaike, Information Criterion (CAIC), and Bayesian
Information Criterion (BIC).

\subsection{Application 1: Voltage data}
We consider the times of failure and running times for a field-tracking study (Meeker and Escobar, 1998).
We fit the Rayleigh distribution ($\alpha=\beta=1$) to
find initial values for $\theta$ and $\delta$. All computations
are done through the NLMIXED subroutine in SAS. Table \ref{EMV1} lists the MLEs (their standard errors in parentheses), and the previous measures, which reveal that the GOLLG distribution
can be chosen as the best model.

\begin{table}[!htb]
\centering {\caption{Findings for voltage data.  \label{EMV1}}\vspace*{0.3cm}
\begin{tabular}{c|cccc|ccccc}
\hline \hline Model  & $\alpha$ & $\beta$ & $\delta$ & $\theta$  &$\rm{AIC}$ &$\rm{CAIC}$ & $\rm{BIC}$ \\
\hline \hline
GOLLGR &  0.0437  & 0.4611 & 34.6605 & 0.0011 &362.0 & 363.6 & 367.6\\
 & (0.0111)& (01515)& (0.00002)& (3.958E-6)&  & &  \\\hline
OLLGR &  0.0246  & 1 & 37.0100 & 0.0010 &383.2 & 384.1 & 387.4\\
 & (0.0041)& & (0.0070)& (0.00006)&  & &  \\\hline
EGR &  1  & 0.0133 & 35.3427 & 0.00019 &366.5 & 367.4 & 370.7\\
 & & (0.0014)& (0.0051)& (0.00003)&  & &  \\\hline
GR &  1  & 1 & -0.5079 & 0.000011 &368.4 & 368.8 & 371.2\\
 & &  & (0.1147)& (5.1E-6)&  & &  \\\hline
\hline
\end{tabular}}
\end{table}

The likelihood ratio (LR) statistics in Table \ref{rv}
indicate that the GOLLGR distribution is the best model
among the others. The histogram and the plots of the estimated
densities in Figure \ref{ex1}(a), and those of the
empirical and estimated survival functions in Figure \ref{ex1}(b), support the previous conclusion.

\begin{table}[!htb]
\caption{LR statistics for voltage data.}
\label{rv}\centering {\vspace*{0.3cm}
\begin{tabular}{c|c|c|c}
\hline\hline
Model & Hypotheses & LR statistic  & $p$-value \\ \hline\hline
GOLLGR vs OLLGR & $H_{0}:\beta=1$ vs $H_{1}: H_{0}\, \mbox{is false}$ & 23.2 & $< $0.00001 \\
GOLLGR vs EGR & $H_{0}:\alpha=1$ vs $H_{1}: H_{0}\, \mbox{is false}$ & 6.5 & 0.01079 \\
GOLLGR vs GR & $H_{0}:\beta=\alpha=1$ vs $H_{1}: H_{0}\, \mbox{is false}$ & 10.4 & 0.0055 \\ \hline\hline
\end{tabular}%
}
\end{table}

\begin{figure}[!htb]
\begin{center}
\includegraphics[height=8.5cm]{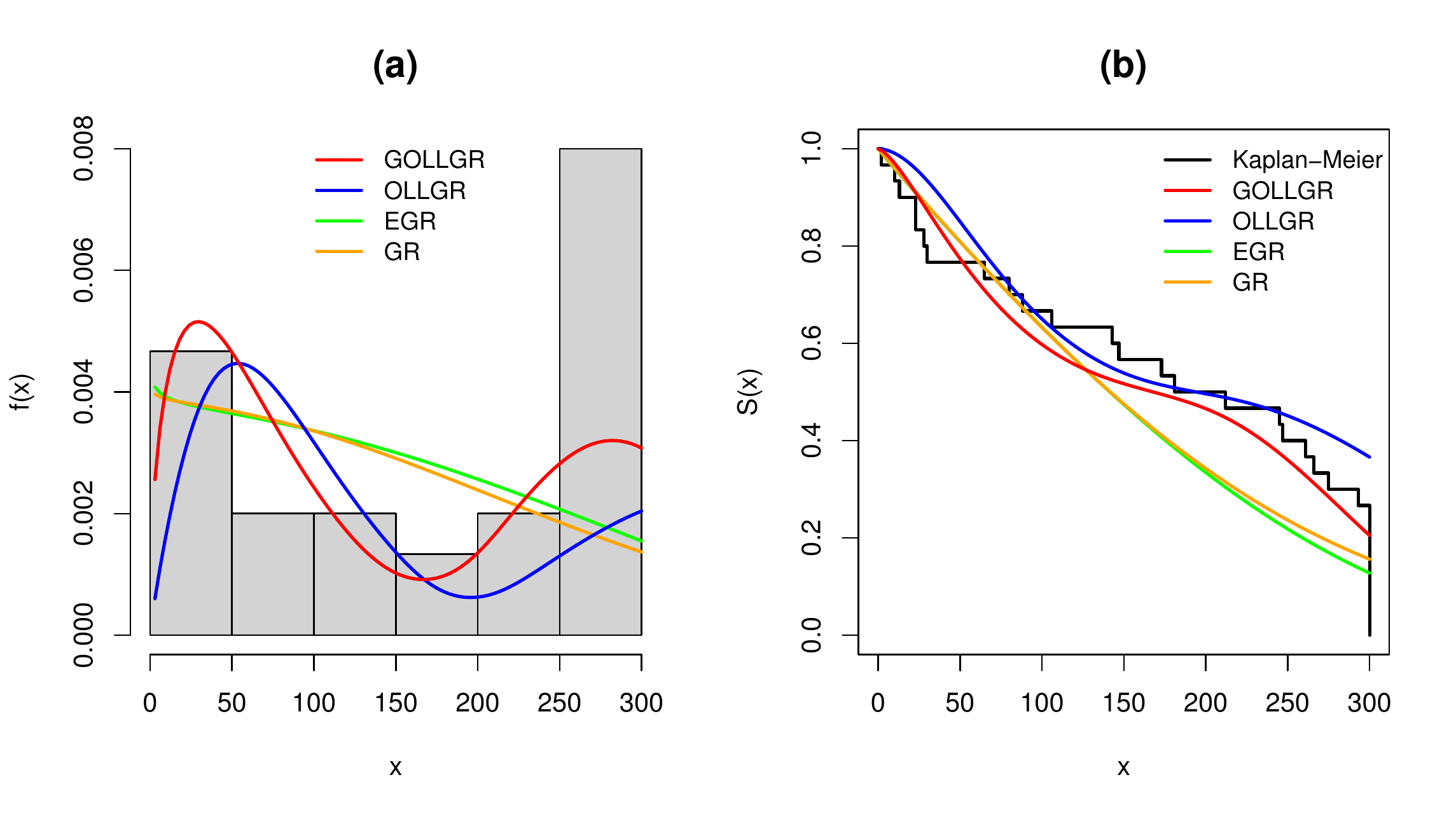}
\end{center}
\caption{Voltage data: (a) Estimated densities. (b) Estimated survival functions and the empirical survival.}
\label{ex1}
\end{figure}

\subsection{Application 2: COVID-19 data}

The second application refers to lifetimes of individuals
diagnosed with COVID-19 (Coronavirus Disease 1999) (Galvão and Roncalli, 2021). Since it was declared an international health emergency, many studies have been conducted to obtain information about the clinical, epidemiological and prognostic aspects of the disease; see, for example, Cordeiro {\it et al.} (2021a),
Cordeiro {\it et al.} (2021b) and  Marinho {\it et al.}
(2021).

In Brazil, the epidemiological data
are disclosed by the Health Information System (available in: \url{https://opendatasus.saude.gov.br/en/dataset/srag-2021-e-2022}.
In this analysis, we work with the \texttt{gamlss} package of {\tt R}.

In this study, 881 patients infected by the virus are considered, confirmed by the RT-PCR test method. The participants consisted of hospitalized patients and outpatients living in the city of Campinas (Brazil) in January and February 2021. The survival consisted of the interval between the first symptoms until the date of death due to COVID-19 (failure). Deaths due to other causes or after the is 73.6\%. Equation (\ref{reg_1}) is considered
with factors associated with the highest risk of death. The results are compared with the OLLGR, EGR and GR sub-regressions.

The following variables were considered for each patient $(i=1,\hdots,881)$:
\begin{itemize}
\itemsep-0.3em

\item $x_i$: time until death due to COVID-19 (in days);
\item $\text{cens}_{i}$: censoring indicator (0 = censored, 1 = observed lifetime);
\item $v_{i1}:$ age (in years);
\item $v_{i2}:$ diabetes mellitus (0= no or not reported, 1= yes).
\end{itemize}

The total number of patients suffering from the comorbidity diabetes was 264 (29.97\%), among whom 104 (39.39\%) died. In turn, of the 617 patients (70.03\%) without the disease or who did not report it, 128 (20.75\%) died. Figure  \ref{surv} presents the Kaplan-Meier survival curve, showing the greater risk of death among patients suffering from diabetes.
\begin{figure}[!htb]
\begin{center}
\includegraphics[height=8cm]{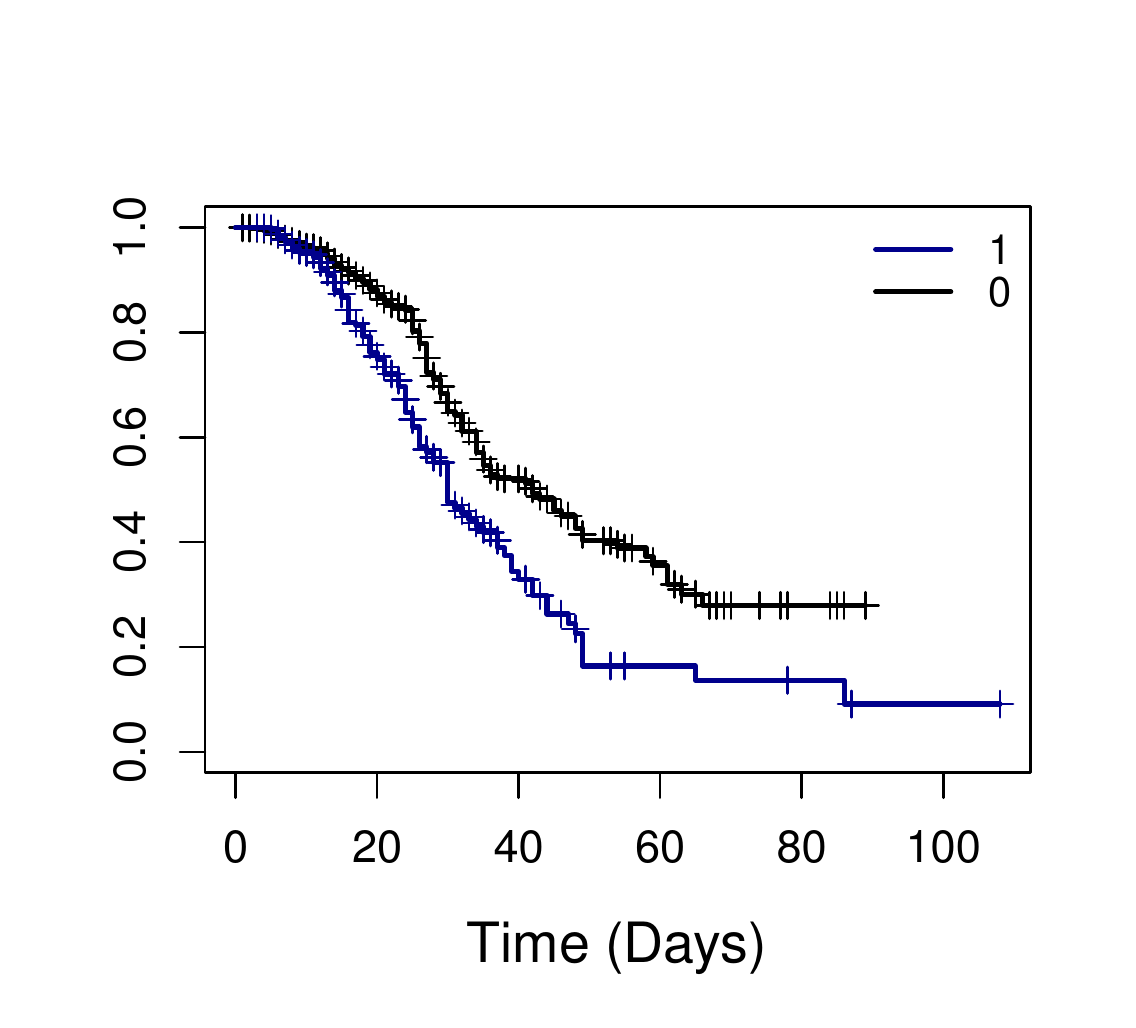}
\end{center}
\caption{Kaplan-Meier survival curve for the variable diabetes mellitus ($1=$ yes, $0=$ no or not informed).}
\label{surv}
\end{figure}

The statistics in Table \ref{aic2} support that the GOLLGR regression can be chosen as the best model. Further,
the LR statistics in Table \ref{lr2} indicate that
the wider regression yields the best fit. Table \ref{mle2} reports the MLEs (SEs in parentheses) from the fitted GOLLGR regression.

\begin{table}[!htb]
\centering
\caption{Findings from the fitted regressions to COVID-19 data.} \label{aic2} \centering  \vspace*{0.3cm}
\begin{tabular}{rrrr}
  \hline\hline
Model & AIC & BIC & CAIC \\
  \hline \hline
GOLLGR & 2222.54 & 2260.78 & 2237.66 \\
  OLLGR  & 2237.73 & 2271.19 & 2262.63 \\
  EGR & 2240.17 & 2273.63 & 2265.07 \\
  GR & 2238.66 & 2267.34 & 2273.34 \\
   \hline\hline
\end{tabular}
\end{table}

\begin{table}[!htb]
\caption{LR statistics for COVID-19 data.}
\label{lr2}\centering {\vspace*{0.3cm}
\begin{tabular}{c|c|c|c}
\hline\hline
Model & Hypotheses & LR statistic & $p$-value \\ \hline\hline
GOLLGR vs OLLGR & $H_{0}:\beta=1$ vs $H_{1}: H_{0}\, \mbox{is false}$ & 17.1915 & $<$0.00001 \\
GOLLGR vs EGR & $H_{0}:\alpha=1$ vs $H_{1}: H_{0}\, \mbox{is false}$ & 19.6298 & $<$0.00001 \\
GOLLGR vs GR & $H_{0}:\beta=\alpha=1$ vs $H_{1}: H_{0}\, \mbox{is false}$ & 20.1202 & $<$0.00001 \\ \hline\hline
\end{tabular}%
}
\end{table}

\begin{table}[!htb]
\centering
\caption{Results from the fitted GOLLGR regressions to COVID-19 data.} \vspace*{0.3cm}
\label{mle2}
\begin{tabular}{rrrrr}
  \hline \hline
 & MLE & SE &   $p$-value \\
\hline \hline
$\lambda_{10}$ & -0.3599 & 0.0289 &  $<$0.0001 \\
$\lambda_{11}$& -0.0028 & 0.0006 & $<$0.0001 \\
  $\lambda_{12}$& 0.0468 & 0.0359 & 0.1920 \\
  $\lambda_{20}$& -10.1148 & 0.1176 & $<$0.0001 \\
  $\lambda_{21}$ & 0.0480 & 0.0022 & $<$0.0001 \\
  $\lambda_{22}$& 0.3331 & 0.0774 & $<$0.0001 \\
  log($\alpha$) & -0.9748 & 0.0130 &  \\
  log($\beta$) & 2.0127 & 0.0130 &   \\
   \hline \hline
\end{tabular}
\end{table}

Figure \ref{res} provides the graphs of the quantile residuals (qrs) (\ref{resq}). The residual index plot (Figure \ref{res}a) reveals that the qrs have a random behavior and that only four observations are outside the $[-3,3]$ range. The normal
probability plot for the qrs (Figure \ref{res}b) indicates
that the residuals follow approximately a normal distribution, which support the fitted regression. Thus, there is no
evidence against the GOLLGR regression assumptions.

\begin{figure}[!htb]
\begin{center}
\includegraphics[width=6cm,height=6cm]{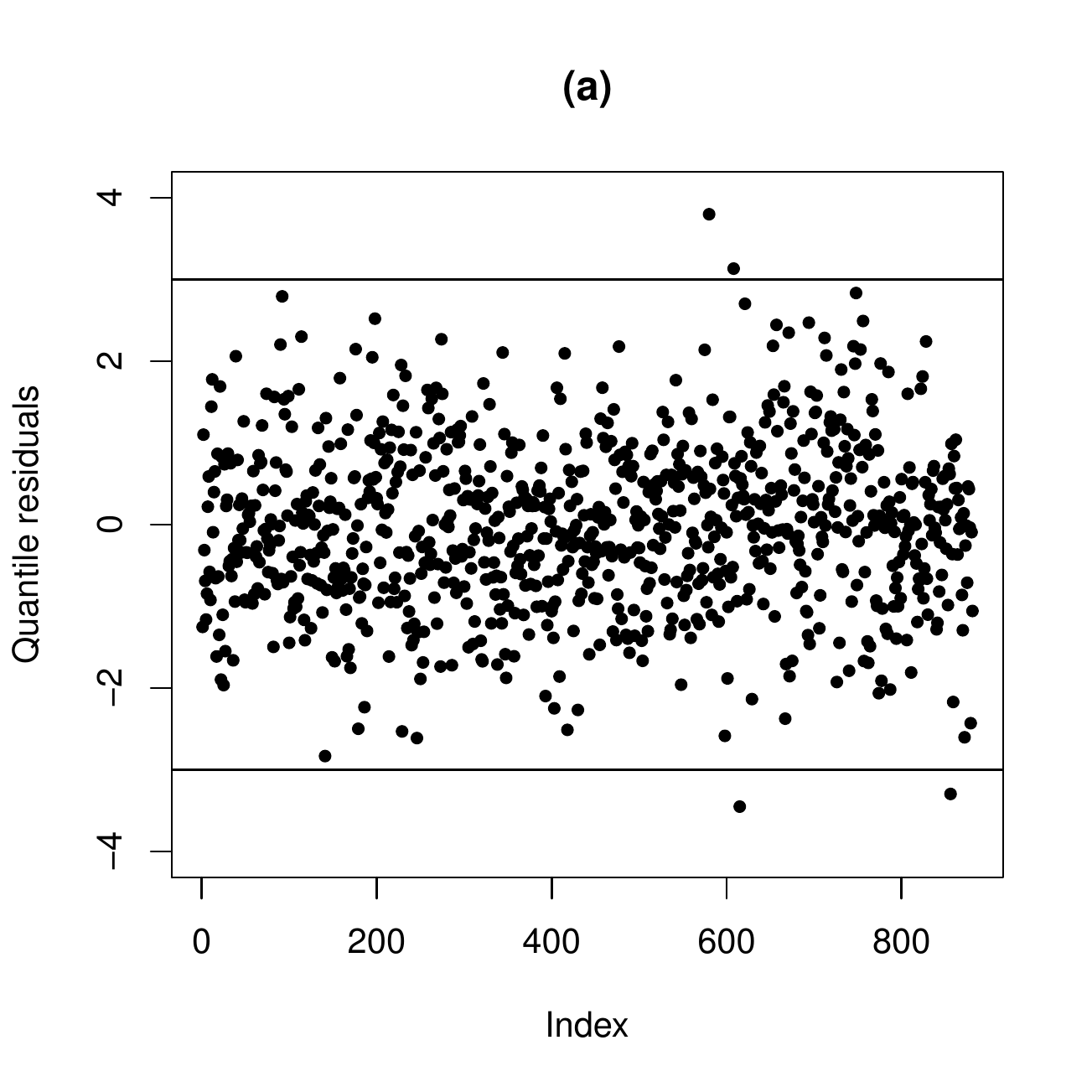} ~~%
\includegraphics[width=6cm,height=6cm]{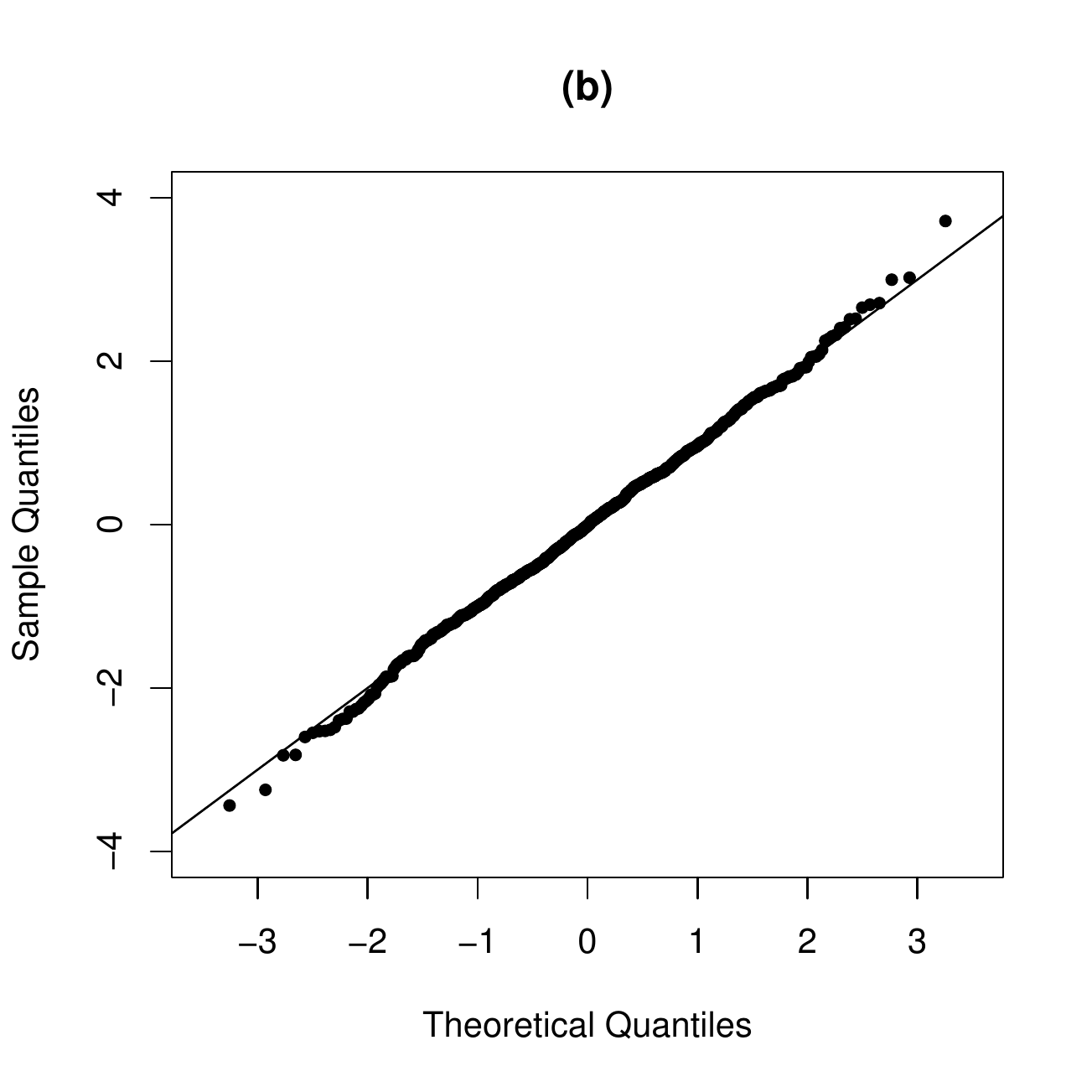}
\end{center}
\caption{COVID-19 data: (a) Index plot of the qrs. (b) Normal probability plot of the qrs.}
\label{res}
\end{figure}

Some interpretations are in order. Table \ref{mle2} shows that the covariable age is significant, meaning that older individuals tend to have a progressively shorter period until death due to this
coronavirus. It is noted a significant difference between individuals with and without diabetes mellitus in relation to
the time until death by COVID-19.

\section{Conclusions}\label{sec:conclusoes}

There is a clear need for extended well-known
distributions and their successful applications in several
areas. The Rayleigh distribution plays a crucial role in modelling and analyzing lifetime data, and several extensions of this
distribution have been published in recent years.
We constructed a new regression based on the four-parameter extended Rayleigh distribution, and showed its utility in the analysis of lifetime data.

\end{document}